\documentclass[letterpaper]{article} 
\usepackage{aaai24}  
\usepackage{times}  
\usepackage{helvet}  
\usepackage{courier}  
\usepackage[hyphens]{url}  
\usepackage{graphicx} 
\usepackage{natbib}  
\usepackage{caption} 
\usepackage{algorithm}
\usepackage{algorithmic}
\usepackage{color}
\usepackage{xcolor}
\usepackage{soul, kotex}
\usepackage{multirow}
\usepackage{amsthm}
\usepackage{amsmath}
\usepackage{makecell}
\usepackage{newfloat}
\usepackage{listings}
\usepackage{xspace, comment}
\newcommand{\sysname}{Blind-Touch\xspace}

\newtheorem{theorem}{Theorem}
\usepackage[skip=0.2\baselineskip]{caption}

\usepackage{newfloat}
\usepackage{listings}
\DeclareCaptionStyle{ruled}{labelfont=normalfont,labelsep=colon,strut=off} 
\floatstyle{ruled}
\newfloat{listing}{tb}{lst}{}
\floatname{listing}{Listing}
\title{\sysname: Homomorphic Encryption-Based Distributed Neural Network Inference for Privacy-Preserving Fingerprint Authentication}
\author{
    Hyunmin Choi\textsuperscript{\rm 1,\rm 3}, 
    Simon S. Woo\textsuperscript{\rm 2,\rm 3},
    Hyoungshick Kim\textsuperscript{\rm 3}
}
\affiliations{
    \textsuperscript{\rm 1}NAVER Cloud, South Korea\\
    \textsuperscript{\rm 2}Department of Artificial Intelligence, Sungkyunkwan University, South Korea \\
    \textsuperscript{\rm 3}Department of Computer Science and Engineering, Sungkyunkwan University, South Korea\\

    hyunmin.choi@\{navercorp.com, g.skku.edu\}, swoo@g.skku.edu, hyuong@skku.edu
}

\usepackage{bibentry}
   
\begin{document}

\maketitle

\begin{abstract}
Fingerprint authentication is a popular security mechanism for smartphones and laptops. However, its adoption in web and cloud environments has been limited due to privacy concerns over storing and processing biometric data on servers. This paper introduces \sysname, a novel machine learning-based fingerprint authentication system leveraging homomorphic encryption to address these privacy concerns. Homomorphic encryption allows computations on encrypted data without decrypting. Thus, \sysname can keep fingerprint data encrypted on the server while performing machine learning operations. \sysname combines three strategies to efficiently utilize homomorphic encryption in machine learning: (1) It optimizes the feature vector for a distributed architecture, processing the first fully connected layer (FC-16) in plaintext on the client side and the subsequent layer (FC-1) post-encryption on the server, thereby minimizing encrypted computations; (2) It employs a homomorphic encryption-compatible data compression technique capable of handling 8,192 authentication results concurrently; and (3) It utilizes a clustered server architecture to simultaneously process authentication results, thereby enhancing scalability with increasing user numbers. \sysname achieves high accuracy on two benchmark fingerprint datasets, with a 93.6\% F1-score for the PolyU dataset and a 98.2\% F1-score for the SOKOTO dataset. Moreover, \sysname can match a fingerprint among 5,000 in about 0.65 seconds. With its privacy-focused design, high accuracy, and efficiency, \sysname is a promising alternative to conventional fingerprint authentication for web and cloud applications.
\end{abstract}

\section{Introduction}

\emph{Fingerprint authentication} is a biometric method that uses the unique characteristics of an individual's fingerprint. It is favored for smartphones and laptops because of its security and convenience \cite{de2015feel,mare2016study,Lovisotto_2020_CVPR_Workshops,cho2020security}. 
However, the adoption of fingerprint authentication in web and cloud environments faces challenges due to the risk of unauthorized access to sensitive biometric data on servers~\cite{rui2018survey}. For example, in June 2015, the US Office of Personnel Management suffered a security breach in which over 5.6 million fingerprint records were stolen, highlighting the dangers of storing biometric data remotely~\cite{gootman2016opm}.

In web and cloud environments, homomorphic encryption (HE)~\cite{gentry2009fully, acar2018survey} is a promising solution for privacy-sensitive applications such as fingerprint authentication. HE allows computations on encrypted data without decryption on a server. Although many research works have adopted this for fingerprint authentication~\cite{b24:kim2020efficient, b25:yang2020secure}, they face challenges due to the significant computational overhead under HE. These challenges primarily arise from complex minutiae representations in fingerprints, which are computationally demanding and prone to variations that can affect the matching process. To overcome these limitations, Engelsma et al.~\cite{engelsma2019learning} introduced a deep learning-based authentication technique using HE. However, even with its enhanced speed, an average authentication time of 3.4 seconds to search among 5,000 fingerprints (considering only the feature vector's encryption time) is not practical for real-world services. Implementing a conventional convolutional neural network (CNN) with HE is challenging due to the costly operations required in the convolution and pooling layers. Specifically, convolution layers necessitate rotations and multiplications for filter application, and pooling layers require a substantial number of multiplications. Additionally, encrypted feature vectors are significantly larger than plaintext feature vectors; therefore, we must do our best to reduce the size of the feature vector, which can consequently lead to a decrease in model accuracy.

To overcome this challenge, we introduce a novel privacy-preserving fingerprint authentication system, \sysname. \sysname employs a distributed deep learning architecture involving clients and a server (see Figure~\ref{figure0}). Clients handle feature extraction through CNN operations on plaintext data while the server performs searching tasks to identify the most suitable fingerprint match with the encrypted feature vector. We optimize the feature vector size to 16, significantly less than the 192 features used in \textit{DeepPrint}~\cite{engelsma2019learning}, by processing the first fully connected layer (FC-16) in plaintext on the client and the second fully connected layer (FC-1) post-encryption on the server. Next, to further enhance performance, we propose a novel compression method compatible with HE to process 8,192 authentication results concurrently. Finally, we implement a clustered architecture to process authentication results simultaneously on multiple servers. To demonstrate the feasibility of \sysname, we built and provided a fully functional cloud-based fingerprint authentication system (\url{https://github.com/hm-choi/blind-touch}).

\begin{figure}[t]
\centerline{\includegraphics[width=1\columnwidth]{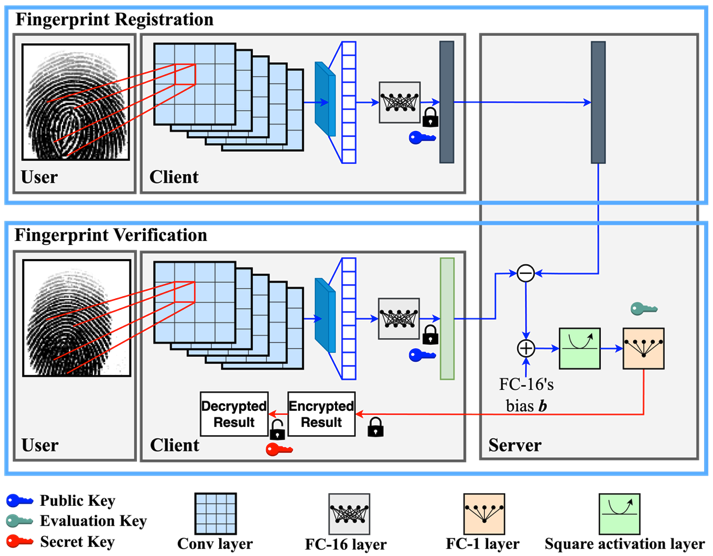}}
    \caption{Overview of \sysname. FC-$N$ refers to a fully connected layer with an output size of $N$.}
    \label{figure0}
\end{figure}

Our main contributions are summarized as follows:

\begin{itemize}
\item \textbf{Design and implementation of \sysname.} We develop a practical distributed HE-based fingerprint authentication system that facilitates efficient and precise neural network inference under HE.

\item \textbf{Demonstration of \sysname's superior recognition accuracy and processing time.} \sysname achieves a 93.6\% F1-score on the PolyU fingerprint dataset~\cite{lin2018matching} and a 98.2\% F1-score on the Sokoto fingerprint dataset~\cite{b1:shehu2018sokoto}. Additionally, it maintains an average search time of 650 milliseconds to identify a match within a pool of 5,000 fingerprints.

\item \textbf{Formal security analysis of \sysname.} We provide a formal security analysis demonstrating that no probabilistic polynomial-time adversary can extract information about the encrypted fingerprint features processed by \sysname in the chosen plaintext attack (IND-CPA) threat model. 

\end{itemize}

\section{Background} 

\subsection{Homomorphic Encryption (HE)} 

Homomorphic encryption (HE) is an encryption scheme that allows third parties, such as cloud service providers, to perform computations on encrypted data without decryption. If $m_1$ and $m_2$ are messages, $Enc$ denotes the homomorphic encryption function, and $f$ and $f'$ represent computationally feasible functions for ciphertext and plaintext inputs, then homomorphic encryption ensures $f(Enc(m_1), Enc(m_2)) = Enc(f'(m_1, m_2))$.

Researchers have introduced various homomorphic encryption algorithms over the years. For instance, Brakerski, Gentry, and Vaikuntanathan introduced the BGV algorithm \cite{b9:brakerski2014leveled, b10:brakerski2012fully, b11:fan2012somewhat}, supporting integer-based arithmetic operations like addition and multiplication. 
The more recent CKKS algorithm \cite{b13:cheon2017homomorphic, b14:cheon2018bootstrapping} supports floating point number-based arithmetic operations, making it more compatible with statistical and deep learning algorithms \cite{b15:clet2021bfv}. Thus, we selected the CKKS scheme for our \sysname.

The CKKS method is a public-key encryption technique involving a public and a secret key. The public key includes an encryption key for encrypting floating point vectors and an evaluation key for homomorphic operations on ciphertexts. The secret key is reserved for decryption. 

\subsection{Siamese Neural Network with CNNs} \label{sec:Siamese Neural Network with CNNs}

The Siamese neural network, a deep learning structure, measures the similarity between two inputs using two identical sub-networks. Commonly used in image-matching tasks such as fingerprint authentication~\cite{b23:chowdhury2020can} and face authentication~\cite{song2019occlusion, wu2017face}, this network generally comprises two or more CNN models as sub-networks with shared weights in the combining layer. The similarity between two input images is determined using the difference between the feature vectors computed from the sub-networks. \sysname employs a Siamese neural network to compute the difference between individuals' fingerprint images.

\subsection{Fingerprint Authentication with HE}
HE has been regarded as an effective means for confidentially storing and processing sensitive data on the server. While prior methods using HE for fingerprint authentication~\cite{b24:kim2020efficient, b25:yang2020secure} have employed basic filter-based models for streamlined processing, they did not achieve the accuracy performance compared to the state-of-the-art CNN-based techniques with plaintext images. Incorporating a CNN model into HE poses significant challenges due to computationally expensive operations such as convolution layers under HE. 

Recent works~\cite{gilad2016cryptonets, ndss2023} have shown the potential of integrating HE with deep neural networks to devise both efficient and privacy-conscious machine learning models. Dowlin et al.~\cite{gilad2016cryptonets} introduced CryptoNets, which showcases the adaptability of neural networks to encrypted data and underscores the alterations required for compatibility with HE. Similarly, Folkerts et al.~\cite{ndss2023} proposed a framework that expanded the design of HE-driven private machine learning inference. However, their approaches fall short of accommodating CNN models with floating point parameters, which are essential for biometric verification.  

Also, Engelsma et al.~\cite{engelsma2019learning} presented DeepPrint, a deep learning-based fingerprint authentication technique that uses a fixed-length representation. DeepPrint aligns the input fingerprint, extracts a 192-dimensional texture and minutiae combination, and compresses it from floating point numbers to a 200-byte integer format. However, its 3.4-second average authentication time for 5,000 fingerprints is not viable for real-world web or cloud platforms that require quick verification. In \sysname, we address such challenges by employing the CKKS scheme to accommodate a CNN with floating point parameters. We also incorporate the following three different strategies to reduce the computational and storage overheads of HE: distributed architecture, data compression, and cluster architecture. These techniques make \sysname a practical solution for real-world fingerprint authentication.

\section{Overview of \sysname}

We present \sysname, a fingerprint authentication system that facilitates efficient inference using a Siamese neural network under HE in a distributed setup (see Figure \ref{figure0}). To reduce the computational load of HE, we optimized the network architecture, placing all convolutional layers on the client side, with only a fully connected layer and a square activation layer on the server side. The client captures a fingerprint and processes it using CNN operations and a fully connected layer with an output size of 16 (FC-16), then encrypts the resulting feature vector with the public key to preserve the confidentiality of the raw feature vector. We minimized the feature vector size while retaining the accuracy of fingerprint authentication, thus maximizing the number of feature vectors that can be stored in a single ciphertext. The size of 16 for the feature vector enables the storage of up to 512 fingerprint feature vectors in a single ciphertext with 8,192 slots, allowing for simultaneous comparison of an individual's fingerprint feature information against 512 registered users. The encrypted vector is then sent to the server for verification. Upon receipt, the server computes the difference between the incoming encrypted feature vector and each registered user's stored encrypted vector. The server computes the encrypted feature vector using the evaluation key. A fully connected layer with an output size of 1 (FC-1) and a square activation layer are subsequently applied to this differential input. The resulting values are relayed back to the client. Utilizing the sigmoid operation, the client decrypts the received data with the secret key and determines the highest match probability for a registered user's feature vector. By comparing this probability against a predetermined threshold, the client checks whether the provided fingerprint image matches the registered user's image. Theorem \ref{thm:theorem2} demonstrates that this design yields equivalent results to processing FC-16 and FC-1 post-encryption.

\sysname's sub-CNN network has a feature size of 16, derived from the FC-16 layer. Originally, the sub-CNN network had five CNN layers, with the FC-16 layer following the subtraction. In this setup, the feature vector's size could increase to 25,088. However, we discovered the FC-16 layer can be computed before encryption. Theorem \ref{thm:theorem2} proves that for any linear function \( f \) defined as \( f(x) = xA+b \) and any homomorphic function \( h \), the following is true:
\[ f(h(x_{1}) - h(x_{2})) = h(f(x_{1})) - h(f(x_{2})) + b \]
The function \( f \) can be applied to the encrypted data before decryption. This implies that the FC-16 layer can be computed pre-encryption, reducing the feature vector's size to 1. After subtraction, the bias \( b \) of the FC-16 layer is added to the encrypted data. Consequently, \sysname achieves the same accuracy as the original configuration while reducing computational cost and data encryption requirements.

\begin{theorem}
\label{thm:theorem2}
Let \( f \) be a linear function defined as \( f(x) = xA + b \) and \( h \) be a homomorphic function. For any \( x_{1} \), \( x_{2} \), and floating-point vector \( b \), the following holds:
\[ f(h(x_{1}) - h(x_{2})) = h(f(x_{1})) - h(f(x_{2})) + b \]
\end{theorem}

\begin{proof}
Starting from the left side:
\begin{align*}
f(h(x_{1}) - h(x_{2})) &= (h(x_{1}) - h(x_{2}))A + b \\
&= h(x_{1})A - h(x_{2})A + b \\
&= h(x_{1}A + b) - h(x_{2}A + b) + b \\
&= h(f(x_{1})) - h(f(x_{2})) + b
\end{align*}
\end{proof}

For training, we utilize a publicly available fingerprint image dataset. Using pairs of same-user and different-user fingerprints from this dataset, \sysname can be trained to accurately identify and match fingerprints. Notably, this training is conducted using plaintext images as we employ publicly available fingerprint images, not specific individuals' private fingerprint data. Once trained, the network is repurposed for authentication of registered users.

Next, we describe fingerprint registration and authentication procedures in detail. 

\subsection{Key Generation and Distribution}
A system administrator oversees multiple client devices with fingerprint scanners. Utilizing a key generator, the administrator generates four unique keys grouped into three types: the \emph{public key} for encrypting fingerprint data; the Galois key for rotating ciphertext; the relinearization key for reducing ciphertext size after multiplication; and the \emph{secret key} for data decryption. This paper terms the Galois and relinearization keys collectively as the \emph{evaluation key}. The administrator securely embeds the public and secret keys on client devices. Subsequently, the public and evaluation keys are relayed to authentication servers via a secure channel. Administrators can then efficiently establish keys on their managed devices and register the public keys with the server, adhering to a standard key setup protocol.

\subsection{Fingerprint Registration}
When a user (denoted by $u$) registers her fingerprint with \sysname, it begins by extracting and processing the feature vector of the fingerprint image using the CNN model's layers on the client side. The final layer of the CNN model on the client side is a fully connected layer (FC-16), which generates a 16-element feature vector denoted by $<u_1, u_2, \cdots, u_{16}>$. We encrypt the feature vector using the client's secret key to protect the feature vector from the server. During fingerprint registration, the client creates a ciphertext $C_u$ containing the user's encrypted feature vector in its first 16 elements. The remaining space in the ciphertext is filled with zeros, as illustrated in Figure~\ref{figure7}. Finally, the client sends $C_u$ to the server, along with the user's unencrypted identity information ${ID}_{u}$.

\begin{figure}[htb!]
\centerline{\includegraphics[width=1\columnwidth]{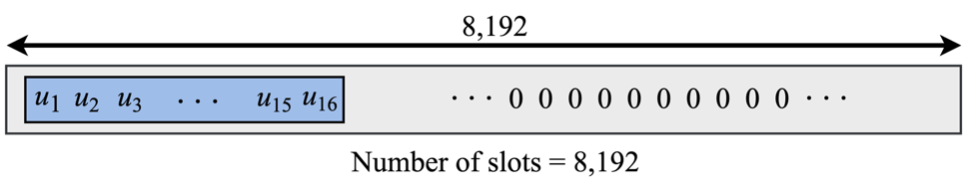}}
\caption{Ciphertext $C_u$ containing the user $u$'s encrypted feature vector for fingerprint registration.}
\label{figure7}
\end{figure}

In \sysname, the server stores encrypted feature vectors of registered users in a dedicated ciphertext, denoted as $C_r$. Since each feature vector has a size of 16, and $C_r$ can hold up to 8,192 available slots, $C_r$ can store a maximum of 512 feature vectors. When a user $u$ registers, the server searches for 16 consecutive empty slots in $C_r$ to store her encrypted feature vector. Assuming that the number of registered users is less than or equal to 512, the server can always find sufficient empty slots to store the user u's feature vector. Suppose $m$ users have already registered, and $u$ is the ($m$+1)th user. To store $u$'s encrypted feature vector in $C_r$, the server first stores the unencrypted user identity information ${ID}_u$ in a plaintext database with the index $m$+1. Then, the server rotates $C_u$ by ($m$+1)×16 positions to the right to align it with the empty slots in $C_r$. This is done to ensure that the encrypted feature vectors of all registered users are stored in contiguous blocks in $C_r$. Finally, the rotated $C_u$ is added to $C_r$, as illustrated in Figure~\ref{figure9}.

\begin{figure}[t]
\centerline{\includegraphics[width=1\columnwidth]{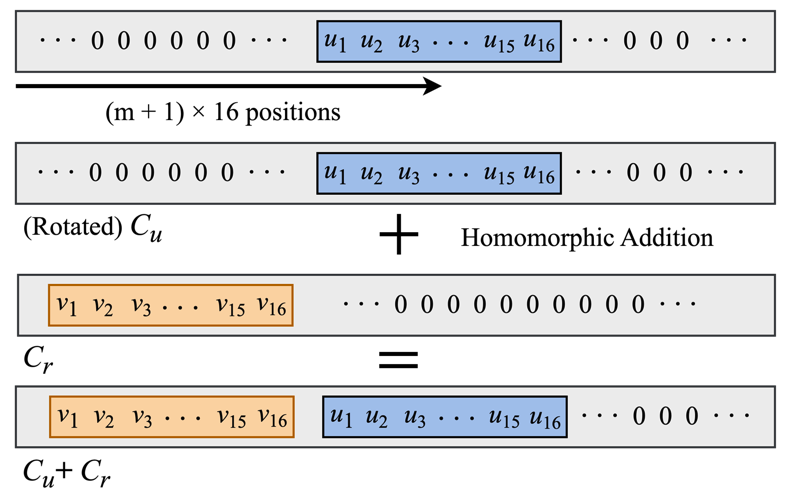}}
\caption{Addition of the rotated $C_u$ and $C_r$.}
\label{figure9}
\end{figure}

\subsection{Fingerprint Authentication}
\label{sec:Fingerprint Authentication}

When user $u$ attempts to authenticate using a new fingerprint image, the feature vector of this image is extracted and processed on the client side through the CNN model's layers, mirroring the fingerprint registration phase. This yields a 16-element feature vector, $<\hat{u}_1, \hat{u}_2, \cdots, \hat{u}_{16}>$, from the FC-16 layer of the CNN model. For encryption, the client produces a ciphertext $C_u$ comprising the user's encrypted feature vector $<\hat{u}_1, \hat{u}_2, \cdots, \hat{u}_{16}>$, replicating this vector 512 times. The client then forwards $C_u$ to the server for authentication.

The server-side computations are sequentially represented in Figure \ref{figure10} as follows: (1) Upon receiving $C_u$, the server attempts to compare $C_u$ with $C_r$; (2) the server performs the subtraction of $a = C_r - C_u$; (3) the server computes the addition of the bias $b$ to the $a$; 
(4) the server performs the square function necessitating a multiplication; (5) the server conducts FC-1 layer operations necessitating a multiplication. Note that traditionally executing the fully connected layer operations would entail iterative rotations and additions after multiplying the result of the previous layer with the coefficients of the FC-1 layer. To avoid these multiplications, \sysname repeatedly puts the identical feature vector $<\hat{u}_1, \hat{u}_2, \cdots, \hat{u}_{16}>$ into $C_r$. Consequently, this authentication method requires only two multiplications.

\begin{figure}[t]
\centerline{\includegraphics[width=1\columnwidth]{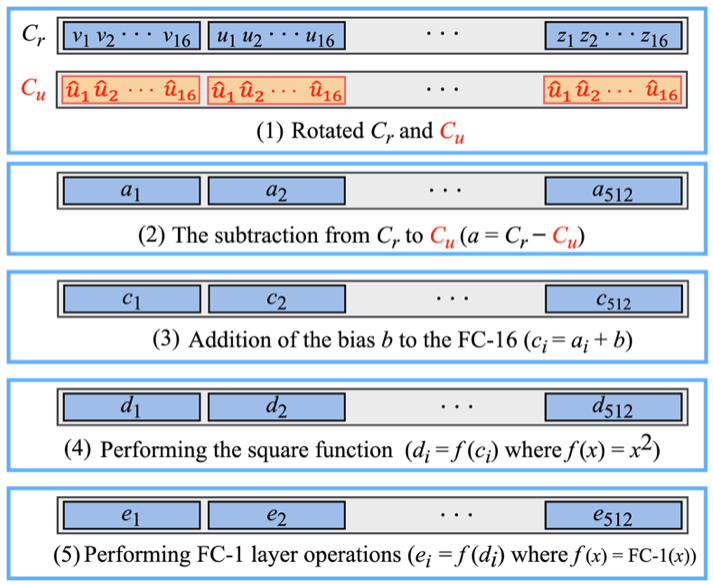}}
\caption{Server-side computations for authentication.}
\label{figure10}
\end{figure}

In \sysname, the $sigmoid$ function is used to compute the probability of matching the feature vector of each registered user with the new fingerprint image's feature vector. However, the $sigmoid$ operation is computationally expensive with HE. To address this challenge, \sysname offloads the $sigmoid$ function from the server to the client. After the server performs the fully connected layer operations, it transmits the results to the client in an encrypted form. The client then efficiently performs the $sigmoid$ operation after decrypting the results. Finally, the client compares the computed probability with a predefined threshold to determine whether the given fingerprint image matches a registered user's fingerprint image. This approach significantly reduces the computational burden on the server, making the authentication process more efficient.

\subsection{Compression Method}
\label{sec:Compression Method}

If the number of registered users ($N$) on the server surpasses 512, without compression, $k (= \lceil N/512 \rceil)$ ciphertexts are necessary to relay the results for all $N$ users because a single ciphertext can only encapsulate the authentication results for a maximum of 512 registered users. This can lead to a significant increase in the authentication time.


To address this issue, we propose a compression method, a new technique that consolidates multiple authentication results into a single ciphertext. To compress the results, the server multiplies a one-hot vector with the first ciphertext. The one-hot vector sets the first element of each registered user's feature vector to 1 and fills the remaining elements with zeros (see Figure~\ref{figure12}). The next ciphertext is rotated by one and added to the resulting ciphertext to create a single ciphertext containing the summation result. By repeating this process with subsequent ciphertexts, we can simultaneously transmit up to 8,192 authenticated results to the client using just one ciphertext. The overall compression process is illustrated in Figure~\ref{figure12}.

\begin{figure}[t]
\centerline{\includegraphics[width=1\columnwidth]{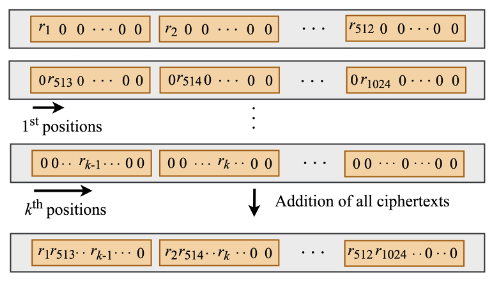}}
\caption{Compression of the authentication results.}
\label{figure12}
\end{figure}

Upon receiving the ciphertext containing the authentication result, the client decrypts and applies the \emph{sigmoid} function to each element of the decrypted vector. The client then looks for the element's index that exceeds the threshold. The corresponding index divided by 16 yields the quotient ($q$) and the remainder ($r$). The index of the original ciphertext stored on the server can be computed as 512$\cdot r$+$q$. This process only requires one additional multiplication and a rotation while significantly reducing the number of ciphertexts transmitted to the client and the authentication time.

\subsection{Cluster Architecture}

The compression method proposed for \sysname facilitates the concurrent processing of up to 512 fingerprint authentications. Consequently, the overall authentication time increases with $\lceil N/512 \rceil$, where $N$ is the number of registered users. If $N$ is considerably large, this can lead to potentially slow authentication times.

One viable strategy to expedite the authentication time in \sysname is to distribute and store the fingerprints across multiple servers and then process them independently and in parallel. We refer to this strategy as the ``\emph{cluster architecture}.'' The proposed cluster architecture is illustrated in Figure~\ref{figure:overall_cluster_architecture}. The server infrastructure consists of two distinct components: the main server and the cluster servers. Upon receiving the encrypted feature vector from the client, the main server relays it to the appropriate cluster servers. Within this structure, the registered (and encrypted) feature vectors are sequentially indexed and stored on each cluster server. 

\begin{figure}[t]
\centerline{\includegraphics[width=1\columnwidth]{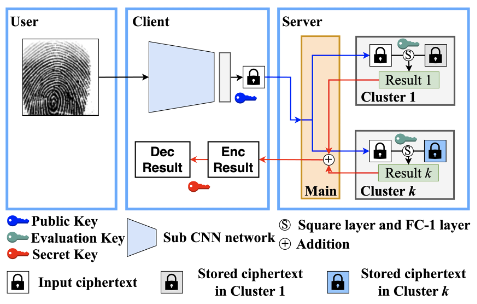}}
\caption{Cluster architecture for \sysname.}
\label{figure:overall_cluster_architecture}
\end{figure}

For illustration, imagine a situation where 1,536 fingerprints are dispersed across three cluster servers. The initial 512 fingerprints would be stored in Cluster 1, the subsequent set, ranging from the 513th to the 1024th fingerprints, would reside in Cluster 2, and the final set would reside in Cluster 3. Each cluster undertakes encrypted fingerprint computations and executes a leftward rotation, ensuring no data overlap. Once computations are complete, the main server collects all resulting ciphertexts into a singular ciphertext and transmits it back to the client. Leveraging this parallel processing mechanism can substantially reduce the authentication time.

\subsection{Implementation of \sysname}
\label{sec:encryption_scheme}

The Siamese network in \sysname consists of two sub-CNN networks that share the same structure and weights. Each CNN network extracts the feature vectors from fingerprint images. The input size of the CNN is $224\times 224$. The CNN architecture is composed of five convolutional layers. Each layer has the same structure except for the size of the output channel. The $i$th convolutional layer has $32 \times i$ number of output channels with the same padding and the size of the stride is one. After the convolutional layer, a BatchNormalization, Swish activation, and MaxPooling layer are applied. To prevent overfitting, we use 40\% dropout during training. The subtracted two CNN architecture is used as an input layer for FC-16. The activation function for FC-16 is a square function, which enables a symmetricity of two CNN networks. Finally, FC-1 and the $sigmoid$ function are used to output the authentication decision.

We adopted the CKKS scheme~\cite{b13:cheon2017homomorphic, b14:cheon2018bootstrapping} because of its efficiency in arithmetic operations over floating-point numbers~\cite{b15:clet2021bfv}. We used the SEAL-Python library (\url{https://github.com/Huelse/SEAL-Python}). In the CKKS scheme, \textit{depth} is predetermined during key configuration. Ciphertext size increases with increasing \textit{depth}. Table~\ref{table:size_of_ciphertext_by_depth}. shows the ciphertext sizes relative to \textit{depth}. To improve \sysname's efficacy, it is important to reduce the ciphertext size, which necessitates minimizing \textit{depth}. We set the \textit{depth} to only support three multiplications, which matches the number of multiplications needed for \sysname.

\begin{table}[t]
\begin{center}
\centering\normalsize
\renewcommand{\arraystretch}{1}
\begin{tabular}{|c|c|c|c|c|c|}
\hline
\emph{depth} & 1 & 2 & 3 & 4 & 5\\
\hline
\hline
Size (KB) & 459 & 658 & 855 & 1,075 & 1,280\\
\hline
\end{tabular}
\end{center}
\caption{Ciphertext sizes according to depth with $d$=16,384 and a log scale factor of 40 in the SEAL-Python library.}
\label{table:size_of_ciphertext_by_depth}
\end{table}

\begin{table}[t]
\begin{center}
\centering\normalsize
\renewcommand{\arraystretch}{1}
\begin{tabular}{|c|c|c|}
\hline
Key and Ciphertext & Size \\
\hline
\hline
Public key & 1.1MB \\
\hline
Galois key & 117MB \\
\hline
Relinearization key & 4.5MB \\
\hline
Secret key  & 559KB \\
\hline
Ciphertext  & 856KB \\
\hline
\end{tabular}
\end{center}
\caption{Sizes of different keys and ciphertext.}
\label{table:size_of_key_pair_and_ciphertext}
\end{table}

Given polynomial modulus degree $d$=16,384, the maximum number of encryptable elements is 8,192 (denoted as \textit{number of slots}). Table~\ref{table:size_of_key_pair_and_ciphertext} presents the sizes of the public key, Galois key, relinearization key, ciphertext, and the number of slots for the parameter settings $d$=16,384, the coefficient modulus is 240, and the log scale factor is 40 using the SEAL-Python library, which are computationally equivalent to 192-bit security in modern symmetric key encryption algorithms~\cite{rahulamathavan2022privacy}. These parameter configurations support up to three multiplications.

\section{Evaluation}
\label{sec:Experiments_and_results}
To evaluate the performance of \sysname, we perform experiments on two publicly available popular benchmark fingerprint datasets: the PolyU Cross Sensor Fingerprint Database~\cite{lin2018matching} and the SOKOTO Conventry fingerprint dataset~\cite{b1:shehu2018sokoto}. Figure~\ref{figure:Set_of_fingerprint_images} shows some sample fingerprint images from each dataset.

\begin{figure}[t]
\centerline{\includegraphics[width=1\columnwidth, height=0.3\linewidth]{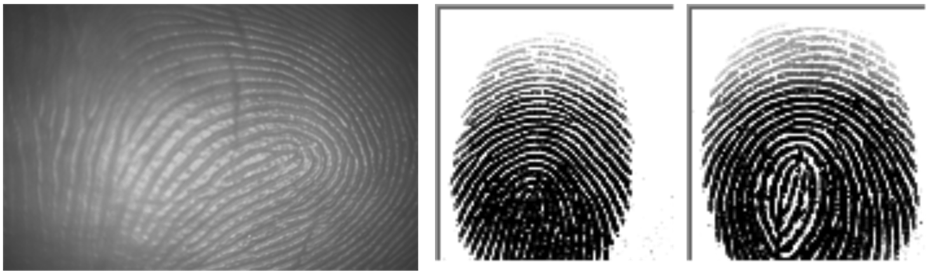}}
\caption{Samples from the PolyU dataset (on the left) and the SOKOTO dataset (on the right).}
\label{figure:Set_of_fingerprint_images}
\end{figure}

\textit{PolyU Cross Sensor Fingerprint Database} (\textbf{PolyU}): This dataset contains contact-based and contactless-2D fingerprint data. We use a processed version of the contactless dataset, which has two sessions. The first session includes 336 subjects, each with 6 fingerprint images, while the second session has 160 subjects, each with 6 images. From the 496 subjects, we choose 296 for training and the remaining 200 for testing, aligning with the dataset configuration in~\cite{feng2023detecting}. The test set consists of 3,000 genuine and 19,900 imposter pairs.

\textit{SOKOTO Coventry Dataset} (\textbf{SOKOTO}): This dataset comprises 6,000 fingerprint images, each measuring 96 $\times$ 103 pixels. We randomly partition this dataset into 3,600 fingerprints for training, 1,200 for validation, and 1,200 for testing. Each fingerprint image has a unique label, so we use one-shot learning~\cite{b5:koch2015siamese}. One-shot learning is a machine learning classification technique that involves training with only one example per class. Given the singular label in the SOKOTO dataset, we employ the pre-processing methods outlined in Table~\ref{table:alternation_for_sokoto_dataset} to create genuine and imposter pairs. To form 1,200 genuine pairs, we apply each pre-processing technique to each fingerprint image. Subsequently, we select 8,400 imposter pairs, maintaining a 1:7 ratio between genuine and imposter datasets, which is similar to the PolyU dataset. Examples of the test samples are illustrated in Figure~\ref{figure:socoto_dataset_pre_processing}.

\begin{figure}[t]
\centerline{\includegraphics[width=1\columnwidth]{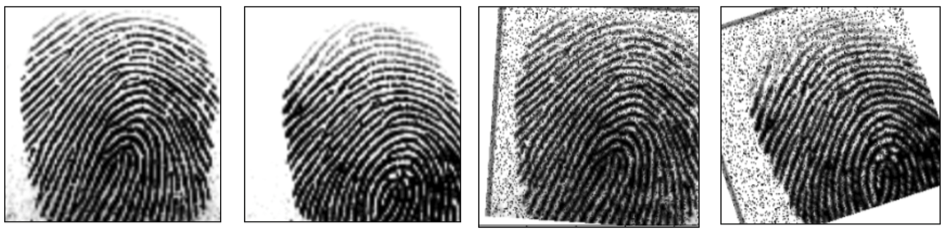}}
\caption{Test samples from the SOKOTO dataset: Original image (on the left) and its corresponding pre-processed versions (on the right).}
\label{figure:socoto_dataset_pre_processing}
\end{figure}

\begin{center}
\begin{table}[b]
\centering\small
\setlength{\tabcolsep}{5pt}
\renewcommand{\arraystretch}{1}
\begin{tabular}{|c|c|c|c|}
\noalign{\smallskip}\noalign{\smallskip}\hline
Dataset & Train & Validation & Test \\
\hline\hline
Dropout (\%) & $0.01 \sim 0.15$ & $0.01 \sim 0.15$ & $0.01 \sim 0.15$ \\
\hline
Scaling (\%) &  $90 \sim 110$ & $90 \sim 110$ & $90 \sim 110$ \\
\hline
Translation (\%) & $-10 \sim 10 $ & $-10 \sim 10 $ & $-10 \sim 10 $  \\
\hline
Rotation ($^{\circ}$) &  $-30 \sim 30$ &$-30 \sim 30$ &$-30 \sim 30$\\
\hline
Gaussian blur  & sigma = 0.7 & sigma = 0.7 & sigma = 0.7 \\
\hline
\end{tabular}
\caption{Pre-processing to generate test samples in the SOKOTO dataset with a dropout rate of 0.5 per channel.}
\label{table:alternation_for_sokoto_dataset}
\end{table}
\end{center}   
 
\subsection{Experimental Setup}

We perform experimental evaluations using the metrics (Accuracy, F1-score, AUC score, and EER score) to understand the performance of \sysname for real-world applications.

The overarching architecture of our system comprises three crucial components: 1 client, 1 main server, and 3 cluster servers. Five servers are used in service, each being a `Standard-g2 Server' sourced from NAVER Cloud. Each server is powered by two cores (Intel(R) Xeon(R) Gold 5220 CPU @ 2.20GHz) and equipped with 8GB of memory. The chosen servers perform similarly to standard personal computers, demonstrating the broader applicability of our study results. We incorporate the Python-based Flask (\url{https://github.com/pallets/flask}) framework to ensure seamless interactivity between these components. 

We choose the hyperparameters (Epoch = 150, Adam optimizer) for optimizing the CNN model in \sysname through experiments.

\subsection{Authentication Accuracy}
In \sysname, the authentication result from the $sigmoid$ function ranges between 0 and 1. A threshold differentiates fingerprints from the same individual and those from different individuals. If the $sigmoid$ outcome exceeds this threshold, the fingerprint images are inferred to be from the same finger. The optimal threshold, dataset-specific and adaptively determined using a validation set, varies. The accuracy and F1-score of \sysname at various thresholds are presented in Table~\ref{table:score_of_polyU_dataset} and~\ref{table:score_of_sokoto_dataset}.

\begin{center}
\begin{table}[t]
\centering\small
\renewcommand{\arraystretch}{1}
\begin{tabular}{|c|c|c|c|c|c|c|}
\noalign{\smallskip}\noalign{\smallskip}\hline
Threshold & 0.2 & 0.1 & 0.05 & 0.02 & 0.01 & 0.005\\
\hline\hline
Accuracy & \textbf{98.3} & 98.3 & 98.2 & 98.0 & 98.0 & 97.6\\
\hline
F1-score & \textbf{93.8} & 93.6 & 93.3 & 92.6 & 92.0 & 91.3\\
\hline
AUC score & 96.8 & 97.1 & 97.3 & 97.4 & 97.5 & \textbf{97.5}\\
\hline
EER score & 5.1 & 4.5 & 3.8 & 3.4 & 2.8 & \textbf{2.5}\\
\hline
\end{tabular}
\caption{Authentication accuracy of \sysname on the PolyU dataset (presented as \%).}
\label{table:score_of_polyU_dataset}
\end{table}
\end{center}

\begin{center}
\begin{table}[t]
\centering\small
\renewcommand{\arraystretch}{1}
\begin{tabular}{|c|c|c|c|c|c|c|}
\noalign{\smallskip}\noalign{\smallskip}\hline
Threshold & 0.2 & 0.1 & 0.05 & 0.02 & 0.01 & 0.005\\
\hline\hline
Accuracy & 99.2 & 98.7 & 99.4 & 99.5 & \textbf{99.5} & 99.2\\
\hline
F1-score & 97.0 & 97.7 & 97.8 & 98.0 & \textbf{98.2} & 97.0\\
\hline
AUC score & 97.8 & 98.6 & 99.0 & 99.2 & \textbf{99.4} & 99.4\\
\hline
EER score & 4.0  & 2.3  & 1.5  & 1.2  & \textbf{0.7}  & 0.8\\
\hline
\end{tabular}
\caption{Authentication accuracy of \sysname on the SOKOTO dataset (presented as \%).}
\label{table:score_of_sokoto_dataset}
\end{table}
\end{center} 

The F1-score of \sysname on the PolyU dataset peaks at 93.8\% with a threshold of 0.2. On the SOKOTO dataset, the best F1-score of 98.2\% is achieved with a threshold of 0.01. However, between thresholds of 0.2 and 0.005, \sysname consistently maintains a high F1-score, exceeding 91\% on the PolyU dataset and surpassing 97\% on the SOKOTO dataset. These results suggest that \sysname's authentication accuracy is not overly sensitive to threshold adjustments and guarantees high authentication accuracy even with a smaller feature vector size.

To demonstrate that \sysname can maintain sufficient authentication accuracy when processing encrypted fingerprint data and remain competitive with state-of-the-art fingerprint authentication solutions designed for plaintext images, we compared the AUC and EER scores on the PolyU dataset against established solutions. Table~\ref{table:score_of_polyU_dataset2} exhibits these comparison results. \sysname achieved an AUC score of 97.5\%, which is only 1.8\% less than the leading results of ContactlessMinuNet~\cite{zhang2021multi} and MinNet~\cite{feng2023detecting}. For EER, \sysname registered 2.5\%, which is 0.6\% higher than the other two methods. Our findings suggest that \sysname serves as a viable alternative for applications sensitive to privacy.

\begin{center}
\begin{table}[t]
\centering\normalsize
\renewcommand{\arraystretch}{1}
\centering\resizebox{\linewidth}{!}{
\begin{tabular}{|c|c|c|}
\noalign{\smallskip}\noalign{\smallskip}\hline
Method & AUC (\%) & EER (\%) \\
\hline\hline
MNIST mindtct~\cite{ko2007user} & 58.9 & 36.9 \\
\hline
\makecell[t]{ MinutiaeNet \\ \cite{nguyen2018robust}} & \multirow{2}{*}{92.0} & \multirow{2}{*}{13.4} \\
\hline
VeriFinger (paid software) & 98.2 & 3.0 \\
\hline
\makecell[t]{ContactlessMinuNet \\ \cite{zhang2021multi}} & \multirow{2}{*}{\textbf{99.3}} & \multirow{2}{*}{\textbf{1.9}} \\
\hline
MinNet~\cite{feng2023detecting} & \textbf{99.3} & \textbf{1.9} \\
\hline
\sysname (Ours) & 97.5 & 2.5 \\
\hline
\end{tabular}}
\caption{Comparison of the AUC and ERR scores for the PolyU dataset with the state-of-the-art fingerprint authentication solutions (in plaintext fingerprints).}
\label{table:score_of_polyU_dataset2}
\end{table}
\end{center} 

\subsection{Rank-1 Accuracy}
To evaluate the feasibility of \sysname for 1:N matching, we measured its Rank-1 accuracy on the distorted SOKOTO dataset, achieving 91\%. This demonstrates \sysname's potential for extension to 1:N matching tasks.
 
\subsection{Execution Time and Storage Performance} 

We analyze the time and storage performance of \sysname in comparison with the state-of-the-art solution, \textit{DeepPrint}~\cite{engelsma2019learning}. We use cosine similarity for \textit{DeepPrint}, following the parameters detailed in~\cite{engelsma2019learning}. Since \textit{DeepPrint} does not provide code to generate feature vectors from fingerprints, we utilize the pre-generated 5,000 feature vectors available in their open-source project. To ensure a fair comparison, we conduct experiments using 5,000 fingerprints extracted from the PolyU dataset, allowing for duplicates.

\begin{center}
\begin{table}[t]
\centering\normalsize
\renewcommand{\arraystretch}{1}
\centering\resizebox{\linewidth}{!}{
\begin{tabular}{|c|c|c|c|c|}
\noalign{\smallskip}\noalign{\smallskip}\hline
Scheme & Input & Enc  & Auth & Dec \\
\hline\hline
\textit{DeepPrint} w/ CKKS & 62  & 2,635 & 772 & 4 \\
\hline
\textit{DeepPrint} w/ BFV & 54 & 1,960 & 4,297 & \textbf{1}\\
\hline
\sysname (Ours)& \textbf{0.8} & \textbf{18} & \textbf{485} & 8\\
\hline
\end{tabular}}
\caption{\sysname vs. \textit{DeepPrint} in input size (MB), encryption (Enc.), authentication (Auth.), and decryption (Dec.) times ($ms$).}
\label{table:time_compare_to_deepprint}
\end{table}
\end{center} 

Table~\ref{table:time_compare_to_deepprint} presents the comparative results between \sysname and two versions of \textit{DeepPrint} (CKKS and BFV). Experimental results show that \sysname substantially outperforms \textit{DeepPrint} in both execution time and storage performance, except for the decryption task. The input vector size for \sysname is approximately 68 times smaller than that of \textit{DeepPrint} w/ BFV. Moreover, encryption and authentication times are about 109 times and 9 times faster, respectively. Even when compared to the relatively faster authentication time of \textit{DeepPrint} w/ CKKS, \sysname is around 1.6 times faster. A drawback of \sysname is its decryption time, which is roughly 8 times and 4 times slower than \textit{DeepPrint} w/ BFV and \textit{DeepPrint} w/ CKKS, respectively. However, an 8 $ms$ decryption duration remains a negligible fraction of the total execution time. When summing up the times for all tasks in \sysname, the total execution time of \sysname is just 511 $ms$ on average. However, the actual authentication duration is longer, extending to approximately 650 $ms$, when the feature extraction and network delivery times are considered. These results underscore that \sysname is still well-suited for practical user authentication services, despite the integration of HE.

\section{Ablation Study}
\label{sec:ablation_study}

An ablation study was conducted to evaluate the impact of key components in \sysname.

\textbf{FC-16 Evaluation after Encryption.} The input size for FC-16 (25,088) is substantially larger than for FC-1 (16), resulting in computational overhead. To optimize the computation overhead of these two FC layers, the first fully connected layer (FC-16) is computed in plaintext on the client, while the second fully connected layer (FC-1) is processed post-encryption on the server. This approach reduces the total evaluation time from 15,096 seconds to just 0.65 seconds.

\textbf{Use of Compression Method.} We evaluate the efficiency of our new compression technique. Consider $N$ as the count of users (or feature vectors) stored on the server. When $N$ surpasses 512, the server conducts $\lceil N/512 \rceil$ similarity checks, as each ciphertext holds a maximum of 512 feature vectors. Without this compression, separate ciphertexts would store each similarity check result, necessitating the server to send back $\lceil N/512 \rceil$ ciphertexts to the client for authentication. However, with our compression method, $\lceil N/512 \rceil$ ciphertexts can be merged into one if $N$ is at most 8,192 since the authentication result occupies just one slot. For $N$ greater than 8,192, ciphertexts reduce to $\lceil N/8,192 \rceil$ instead of $\lceil N/512 \rceil$. This technique keeps ciphertext size constant for up to 8,192 users or feature vectors. Implementing this method, we successfully reduced the ciphertext size from 2.6 MB to 0.26 MB for a 5,000 input test dataset.

\textbf{Use of Clusters.} The cluster architecture considerably impacts \sysname's execution time. In \sysname, a ciphertext can concurrently compare 512 fingerprints using SIMD operations. Hence, for 5,000 registered fingerprints, 10 similarity-matching comparisons, calculated as 10 = $\lceil5,000/512\rceil$, are required. This process is efficiently distributed across clusters for load-balancing. Table~\ref{table:autmentation_time_depending_to_clusters} illustrates the total authentication time for \sysname, varying from 1 to 3 clusters, with 5,000 fingerprints registered.

\begin{table}[b]
\centering\normalsize
\begin{tabular}{|c|c|c|c|}
\noalign{\smallskip}\noalign{\smallskip}\hline
\# of clusters & 1 & 2 & 3 \\
\hline\hline
Time & 1,334.4 & 790.5 & 650.2 \\
\hline
\end{tabular}
\caption{Total authentication time (in $ms$) of \sysname with different cluster counts for 5,000 registered fingerprints.}
\label{table:autmentation_time_depending_to_clusters}
\end{table}

Using a single cluster, \sysname's authentication time is 1,334.4 ms, where one server handles all 10 comparisons. With two clusters, the load is split, with each handling 5 comparisons. In a three-cluster setup, two clusters manage 3 comparisons each, and the third handles 4, optimizing load distribution. Consequently, employing three clusters cuts the operation time by nearly half.

\textbf{Optimization of Feature Vector Size.} The size of the output feature vector impacts the computation time of the HE-based FC layer. The total time for \sysname increased from 0.65 seconds with 16 features to 1.81 seconds with 64 features.

\textbf{Optimization of CNN Model Architecture.} To achieve a balance between performance and the overhead of HE, various neural network depths were tested on the PolyU dataset. The goal was to maintain high authentication accuracy while reducing model complexity. A 5-layer configuration is recommended, as it attained the highest F1 score of 93.8\%, compared to 92.4\% for the 4-layer and 89.3\% for the 6-layer configurations.

\section{Security Analysis} 

We consider a curious server for the adversary model. The server can only see the ciphertext corresponding to the encrypted feature vector of a fingerprint received from the client and the ciphertext corresponding to the encrypted authentication result. In this section, we prove the security of \sysname against a curious server using the simulation-based security in the semi-honest setting, which is widely employed to prove the security of protocols \cite{b30:chandran2022simc}.

We use $\mathcal{C}$ to represent a trusted client and $\mathcal{A}$ to represent an adversarial server. $\mathcal{A}$ wants to obtain information about the user's fingerprint data. For the proof, we generate a simulator $\mathcal{S}$ against the adversary $\mathcal{A}$ as follows.

\medskip\noindent\textbf{The Simulator.} When $\mathcal{C}$ encrypts the feature vector of a fingerprint image using the encryption key, the simulator generates encryption of 0s instead of the embedded feature vector. Note that the simulator has access to the public encryption key, which allows encrypting of any data.

Now, we define the following two games for the \sysname framework $\pi$. 

\begin{itemize}
    \item The game $REAL_{(\pi, \mathcal{C}, \mathcal{A})}$: The client $\mathcal{C}$ encrypts the feature vector extracted from a user's fingerprint and transfers it to the adversary $\mathcal{A}$.
    \item The game $IDEAL_{(\pi, \mathcal{S}, \mathcal{A})}$: The simulator $\mathcal{S}$ encrypts a vector, which is entirely zero-populated, instead of the feature vector of the fingerprint and transfers it to the adversary $\mathcal{A}$.
\end{itemize}

We can prove the computational indistinguishability between $REAL$ and $IDEAL$ games as follows.

\newtheorem{claim}[theorem]{Claim}
\begin{theorem}
\label{theorem:theorem1}
$REAL_{(\pi, \mathcal{C}, \mathcal{A})}$ and $IDEAL_{(\pi, \mathcal{S}, \mathcal{A})}$ are computationally indistinguishable.
\end{theorem}

\begin{proof}
The two games differ only in one aspect is where, in $REAL_{(\pi, \mathcal{C}, \mathcal{A})}$,  $\mathcal{A}$ receives a ciphertext of feature vector of real user's fingerprint and, in $IDEAL_{(\pi, \mathcal{S}, \mathcal{A})}$, $\mathcal{A}$ receives a ciphertext that is generated by encrypting the vector containing all zero elements. Regardless, $\mathcal{A}$ cannot distinguish two ciphertexts computationally because of the indistinguishability chosen plaintext attack (IND-CPA) security~\cite{cheon2020remark} about HE. Both ciphertexts are computationally indistinguishable from the uniform random variable over the ciphertext space under the assumption of the hardness of the ring learning with errors (RLWE) problem~\cite{lyubashevsky2013ideal}. 
\end{proof}

According to the claim, we can conclude that the difference in advantage between these two games is negligible. Therefore, $\mathcal{A}$ cannot obtain any information about the user's fingerprint. Consequently, \sysname achieves security against the curious server. 

Theorem \ref{theorem:theorem1} states that the difference in advantage between the two games is negligible. This means that $\mathcal{A}$ cannot obtain any significant information about the user's fingerprint. Therefore, \sysname achieves security against the curious server.

\section{Conclusion}
As the prevalence of AI services grows, there are growing concerns about data privacy. Therefore, we must not only focus on improving the efficiency of machine learning models but also on ensuring that user data is processed securely and with privacy in mind. We introduce \sysname, a distributed machine learning framework specifically designed for privacy-preserving fingerprint authentication using HE. \sysname's optimized design enables efficient and accurate authentication, even in the presence of HE. Our findings provide essential guidance for the implementation of secure and privacy-preserving AI services.

\section{Acknowledgements}
The authors would thank anonymous reviewers. Hyoungshick Kim is the corresponding author. This work was supported by NAVER Cloud, the Korea Internet \& Security Agency (KISA) grant (No. 1781000003, Development of a Personal Information Protection Framework for Identifying and Blocking Trackers) and the Institute for Information \& communication Technology Planning \& Evaluation (IITP) grants (No. 2022-0-01199, Graduate School of Convergence Security, and No. 2022-0-00688, AI Platform to Fully Adapt and Reflect Privacy-Policy Changes). 
\bibliography{aaai24}

\end{document}